\documentclass{llncs}
\usepackage[T1]{fontenc}
\usepackage[utf8]{inputenc}
\usepackage{rotating}
\usepackage[final]{microtype}
\usepackage{color}
\usepackage{latexsym}
\usepackage{amssymb}
\usepackage{amsmath}
\usepackage{mathtools}
\usepackage{xspace}
\usepackage{multirow}
\usepackage{relsize}
\usepackage{wrapfig}
\usepackage{bussproofs}
\usepackage[hypertexnames=false]{hyperref}
\usepackage[caption=false, position=top]{subfig}
\usepackage{turnstile}
\usepackage{centernot}

\newcommand{\vars}{\mathit{vars}\xspace}
\newcommand{\qbfmodels}{\models_{t}\xspace}
\newcommand{\levels}{\mathit{levels}\xspace}
\newcommand{\eabs}{\mathit{Abs}\xspace}
\newcommand{\unitprop}{\mathrel{\vdash\kern-.65em_{^{_{1}}}}\xspace}
\newcommand{\unitQprop}{\mathrel{\vdash\kern-.75em_{^{_{1\forall}}}}}
\newcommand{\nunitQprop}{\mathrel{\nvdash\kern-.75em_{^{_{1\forall}}}}}
\newcommand{\outercl}{\mathsf{OC}\xspace}
\newcommand{\outerres}{\mathsf{OR}\xspace}
\newcommand{\qior}{\mathsf{QIOR}\xspace}
\newcommand{\newqior}{\mathsf{QIOR}^{+}\xspace}
\newcommand{\qrat}{\mathsf{QRAT}\xspace}
\newcommand{\at}{\mathsf{AT}\xspace}
\newcommand{\qat}{\mathsf{QAT}\xspace}
\newcommand{\newqrat}{\mathsf{QRAT}^{+}\xspace}
\newcommand{\hqspre}{\textsf{HQSpre}\xspace}
\newcommand{\qute}{\textsf{Qute}\xspace}
\DeclareRobustCommand{\qrattool}{\textsf{QRATPre\raisebox{.4ex}{\relsize{-3}{+}}}\xspace}
\newcommand{\textqrate}{\mathsf{QRATE}\xspace}
\newcommand{\textqrata}{\mathsf{QRATA}\xspace}
\newcommand{\textqratu}{\mathsf{QRATU}\xspace}
\newcommand{\textnewqrate}{\mathsf{QRATE}^{+}\xspace}
\newcommand{\textnewqrata}{\mathsf{QRATA}^{+}\xspace}
\newcommand{\textnewqratu}{\mathsf{QRATU}^{+}\xspace}
\newcommand{\texteur}{\mathsf{EUR}\xspace}
\newcommand{\lquplusres}{LQU$^+$\!-resolution\xspace}

\newcommand{\bloqqer}{\textsf{Bloqqer}\xspace}
\newcommand{\aigsolve}{\textsf{AIGSolve}\xspace}

\newcommand{\caqe}{\textsf{CAQE}\xspace}

\newcommand{\rareqs}{\textsf{RAReQS}\xspace}

\newcommand{\depqbf}{\textsf{DepQBF}\xspace}

\newcommand{\UR}{\mathit{UR}}

\newcommand{\var}[1]{\mathsf{var}(#1)}
\newcommand{\prefix}{\Pi}
\newcommand{\clauset}{\psi}
\newcommand{\qclauset}{\phi}
\newcommand{\quant}[2]{\mathsf{Q}(#1,#2)}
\newcommand{\satequiv}{\equiv_{\mathit{sat}}}
\newcommand{\treemodelequiv}{\equiv_{\mathit{t}}}

\newcommand{\FormulasL}{\Phi_{L}}
\newcommand{\FormulasC}{\Phi_{C}}
\newcommand{\prefixL}{\prefix_{L}}
\newcommand{\clausetL}{\clauset_{L}}
\newcommand{\prefixC}{\prefix_{C}}
\newcommand{\clausetC}{\clauset_{C}}

\newcommand{\formulasQUParity}{\Phi_{Q}}
\newcommand{\prefixQUParity}{\prefix_{Q}}
\newcommand{\clausetQUParity}{\clauset_{Q}}

\author{Florian Lonsing \and Uwe Egly}

\institute{Research Division of Knowledge Based Systems \\ 
  Institute of Logic and Computation,  
  TU Wien, Austria \\ 
  \email{firstname.lastname@tuwien.ac.at}
}

\begin{document}

\newcommand{\doctitle}{$\newqrat$: Generalizing $\qrat$ by a More Powerful QBF Redundancy Property}

\title{\doctitle
\thanks{Supported by the Austrian Science Fund (FWF) under grant
  S11409-N23. This article will appear in the \textbf{proceedings} of IJCAR
  2018, LNCS, Springer, 2018.}}

\maketitle

\begin{abstract}
The $\qrat$ (quantified resolution asymmetric tautology) proof system
simulates virtually all inference rules applied in state of the art quantified
Boolean formula (QBF) reasoning tools. It consists of rules to rewrite a QBF
by adding and deleting clauses and universal literals that have a certain
redundancy property. To check for this redundancy property in $\qrat$,
propositional unit propagation (UP) is applied to the quantifier free, i.e.,
propositional part of the QBF.  We generalize the redundancy property
in the $\qrat$ system by QBF specific UP (QUP). QUP extends UP by the
universal reduction operation to eliminate universal literals from clauses. We
apply QUP to an abstraction of the QBF where certain universal quantifiers are
converted into existential ones. This way, we obtain a generalization of
$\qrat$ we call $\newqrat$.  The redundancy property in
$\newqrat$ based on QUP is more powerful than the one in $\qrat$ based on UP. We
report on proof theoretical improvements and experimental results to
illustrate the benefits of $\newqrat$ \mbox{for QBF preprocessing.}
\end{abstract}


\section{Introduction}
\label{sec:intro}

In practical applications of propositional logic satisfiability
(SAT), it is necessary to establish correctness guarantees on the results produced by SAT
solvers by proof checking~\cite{DBLP:journals/cacm/HeuleK17}. The DRAT (deletion resolution asymmetric
tautology)~\cite{DBLP:conf/sat/WetzlerHH14} approach has become 
state of the art to generate and check propositional proofs.

The logic of quantified Boolean formulas (QBF) extends propositional logic by
existential and universal quantification of the propositional
variables. Despite the PSPACE-completeness of QBF
satisfiability checking, QBF technology is relevant in 
practice due to the potential succinctness of
QBF encodings~\cite{DBLP:conf/tacas/FaymonvilleFRT17}.

DRAT has been lifted to QBF to obtain the $\qrat$ (quantified RAT) proof
system~\cite{DBLP:conf/cade/HeuleSB14,DBLP:journals/jar/HeuleSB17}. 
$\qrat$ allows to represent and
check (un)satisfiability proofs of QBFs and compute Skolem function
certificates of satisfiable QBFs. 
The $\qrat$
system simulates virtually all inference rules applied in state of the art QBF
reasoning tools, such as Q-resolution~\cite{DBLP:journals/iandc/BuningKF95}
including its variant long-distance
Q-resolution~\cite{DBLP:conf/sat/KieslHS17,DBLP:conf/iccad/ZhangM02}, and
expansion of universal variables~\cite{DBLP:conf/sat/Biere04a}.

A $\qrat$ proof of a QBF in prenex CNF 
consists of a sequence of
inference steps that rewrite the QBF by adding and deleting clauses and
universal literals that have the \emph{$\qrat$ redundancy
property}. Informally, checking whether a clause $C$ has $\qrat$ amounts to
checking whether all possible resolvents of $C$ on a literal $l \in
C$ (under certain restrictions) are \emph{propositionally implied} by the
quantifier-free CNF part of the QBF. The principle of redundancy checking by
inspecting resolvents originates from the RAT property in propositional
logic~\cite{DBLP:conf/cade/JarvisaloHB12} and was generalized to first-order
logic in terms of \emph{implication modulo
resolution}~\cite{DBLP:conf/cade/Kiesl017}. Instead of a complete  
(and thus computationally hard) propositional implication check on a
resolvent, 
the $\qrat$ system relies on an incomplete check by \emph{propositional unit
propagation} (UP). Thereby, it is checked whether UP can derive the empty
clause from the CNF augmented by the negated resolvent. Hence  
redundancy checking in $\qrat$ is unaware of the quantifier structure, which is entirely ignored in UP.

We generalize redundancy checking in $\qrat$ by making it aware of the quantifier
structure of a QBF. To this end, we 
check the redundancy of resolvents based on \emph{QBF specific UP} (QUP). It
extends UP by the \emph{universal reduction} (UR)
operation~\cite{DBLP:journals/iandc/BuningKF95} and is a polynomial-time 
procedure like UP. UR is central in resolution
based QBF
calculi~\cite{DBLP:conf/sat/BalabanovWJ14,DBLP:journals/iandc/BuningKF95} as
it shortens individual clauses by eliminating  universal literals depending on
the quantifier structure. 
We apply QUP to \emph{abstractions} of the QBF where certain universal quantifiers are
converted into existential ones. The purpose of abstractions is that if a resolvent is found redundant by QUP on
the abstraction, then it is also redundant \mbox{in the original QBF.} 
 
Our contributions are as follows: 
(1) by applying QUP and QBF abstractions
instead of UP, we obtain a \emph{generalization of the $\qrat$ system} which
we call $\newqrat$. In contrast to $\qrat$, redundancy checking in $\newqrat$
is aware of the quantifier structure of a QBF.  
We show that (2) the
redundancy property in $\newqrat$ based on QUP is \emph{more powerful} than
the one in $\qrat$ based on UP. $\newqrat$ can detect redundancies 
which $\qrat$ cannot.
 As a formal foundation, 
we introduce (3) a \emph{theory of QBF abstractions} used in $\newqrat$. 
Redundancy elimination by $\newqrat$ or $\qrat$ can lead to (4) \emph{exponentially
shorter proofs} in certain resolution based QBF calculi, which we point out by
a concrete example. 
Note that here we do not study the power of $\qrat$ or $\newqrat$ as proof
systems themselves, but the impact of redundancy elimination. 
Finally, we report on experimental
results (5) to illustrate the benefits of redundancy elimination by $\newqrat$
and $\qrat$ for QBF preprocessing. Our implementation of $\newqrat$ and
$\qrat$ for preprocessing is the first one reported in the literature. 


\section{Preliminaries}
\label{sec:prelims}

We consider QBFs $\qclauset := \prefix.\clauset$ in \emph{prenex conjunctive
  normal form (PCNF)}    
with a \emph{quantifier prefix} $\prefix := Q_1B_1 \ldots Q_nB_n$ and a
quantifier free CNF $\clauset$ not containing tautological clauses.
The prefix consists of \emph{quantifier blocks} $Q_iB_i$, where $B_i$ are
\emph{blocks} (i.e., sets) of propositional variables and $Q_i \in \{\forall,
\exists\}$ are \emph{quantifiers}. We have $B_i \cap B_j = \emptyset$, 
$Q_i
\not= Q_{i+1}$ and $Q_n = \exists$. 
The CNF $\clauset$ is defined precisely over the variables $\vars(\qclauset)
= \vars(\clauset) := B_1 \cup \ldots \cup B_n$  
in $\prefix$ so that all variables
are quantified, i.e., $\qclauset$ is \emph{closed}.
 The \emph{quantifier}
$\quant{\prefix}{l}$ of literal $l$ is  $Q_i$ if the variable $\var{l}$ of $l$ appears in
$B_i$. 
The set of variables in a clause $C$ is $\vars(C) := \{x \mid l \in C, \var{l} = x\}$.
A literal $l$ is \emph{existential} if $\quant{\prefix}{l} = \exists$ and
\emph {universal}
if $\quant{\prefix}{l} = \forall$.
If $\quant{\prefix}{l} = Q_i$
and $\quant{\prefix}{k} = Q_j$,
then $l \leq_\prefix k$ iff $i \leq j$. We
extend the ordering $\leq_\prefix$ to an arbitrary but fixed ordering
on the variables in every block $B_i$.

An \emph{assignment} $\tau: \vars(\qclauset) \rightarrow
\{\top, \bot\}$ maps the variables  
of a QBF $\qclauset$ to truth constants $\top$ (\emph{true}) or 
$\bot$ (\emph{false}). Assignment $\tau$ is \emph{complete} if it assigns every variable
in $\qclauset$, otherwise $\tau$ is \emph{partial}. By $\tau(\qclauset)$ we denote \emph{$\qclauset$ under
$\tau$}, where each occurrence of variable $x$ in $\qclauset$ is replaced
by $\tau(x)$ and $x$ is removed from the prefix of $\qclauset$, followed by
 propositional simplifications on $\tau(\qclauset)$. 
We consider $\tau$ as a set of literals such that, for
some variable $x$, $x \in \tau$ if $\tau(x) = \top$ and $\bar x \in
\tau$ if $\tau(x) = \bot$.
 
An \emph{assignment tree}~\cite{DBLP:journals/jar/HeuleSB17} $T$ of a
QBF $\qclauset$ is a \emph{complete binary tree} of depth $|\vars(\qclauset)|
+ 1$ where the internal (non-leaf) nodes of each level are
associated with a variable of $\qclauset$. An internal node is universal
(existential) if it is associated with a
universal (existential) variable. The order of variables along every path 
in $T$ respects the extended order $\leq_\prefix$ of
the prefix $\prefix$ of $\qclauset$. An internal node associated
with variable $x$ has two outgoing edges pointing to its children: one
labelled with $\bar x$ and another one labelled with $x$, denoting the
assignment of $x$ to false and true, respectively. Each path $\tau$ in $T$ 
from the root to an internal node (leaf) represents a partial
(complete) assignment. 
A leaf at the end of $\tau$ is labelled by $\tau(\qclauset)$,
i.e., the value of $\qclauset$ under $\tau$. An internal node
associated with an existential (universal) variable is labelled with $\top$
iff one (both) of its children is (are) labelled with $\top$. The QBF
$\qclauset$ is \emph{satisfiable} (\emph{unsatisfiable}) iff the root of $T$ is labelled
with $\top$ ($\bot$).

Given a QBF $\qclauset$ and its assignment tree $T$, a 
subtree $T'$ of $T$ is a \emph{pre-model}~\cite{DBLP:journals/jar/HeuleSB17} of $\qclauset$ if (1) the
root of $T$ is the root of $T'$, (2) for every universal node in $T'$ both
children are in $T'$, and (3) for every existential node in $T'$ exactly one of
its children is in $T'$. A pre-model $T'$ of $\qclauset$ is a \emph{model}~\cite{DBLP:journals/jar/HeuleSB17} 
of $\qclauset$, denoted by $T' \qbfmodels \qclauset$, if each node in $T'$ is labelled with
$\top$. 
A QBF $\qclauset$ is \emph{satisfiable} iff
it has a model. Given a QBF $\qclauset$ and one of its \emph{models} $T'$,
$T''$ is a \emph{rooted subtree} 
of $T'$ ($T'' \subseteq T'$) if $T''$ has the same root as $T'$ and the leaves
of $T''$ are a subset of the leaves of $T'$.

We consider CNFs $\clauset$ defined over a set $B$
of variables without an explicit 
quantifier prefix.  A \emph{model} of a CNF $\clauset$ is a 
model $\tau$ of the QBF $\exists B. \clauset$ which consists only of the
single path $\tau$.   
We write $\tau \models \clauset$ if $\tau$ is a model of $\clauset$. 
For CNFs $\clauset$ and $\clauset'$, $\clauset'$ is \emph{implied} by
$\clauset$ ($\clauset \models \clauset'$) if, for all $\tau$, it holds that if
$\tau \models \clauset$ then 
$\tau \models \clauset'$. Two CNFs $\clauset$ and $\clauset'$ are
\emph{equivalent} ($\clauset \equiv \clauset'$), iff $\clauset \models
\clauset'$ and $\clauset' \models \clauset$. 
We define notation to explicitly refer to QBF models. For QBFs $\qclauset$ and $\qclauset'$, $\qclauset'$ is \emph{implied} by
$\qclauset$ ($\qclauset \qbfmodels \qclauset'$) if, for all $T$, it holds that
if $T \qbfmodels \qclauset$ then $T \qbfmodels \qclauset'$. 
QBFs $\qclauset$ and $\qclauset'$ are
\emph{equivalent} ($\qclauset \treemodelequiv \qclauset'$) iff $\qclauset \qbfmodels
\qclauset'$ and $\qclauset' \qbfmodels \qclauset$, and \emph{satisfiability equivalent} 
($\qclauset \satequiv \qclauset'$) iff
$\qclauset$ is satisfiable whenever $\qclauset'$ is
satisfiable. Satisfiability equivalence of CNFs is defined analogously and
\mbox{denoted by the same symbol '$\satequiv$'.}


\section{The Original $\qrat$ Proof System} \label{sec:trad:qrat}

Before we generalize $\qrat$, we recapitulate the
original proof system~\cite{DBLP:journals/jar/HeuleSB17} and emphasize
that redundancy checking in $\qrat$ is unaware of quantifier structures.

\begin{definition}[\cite{DBLP:journals/jar/HeuleSB17}]
The \emph{outer clause} of clause $C$ on literal $l \in C$ with respect to
prefix $\prefix$ is the clause $\outercl(\prefix, C, l) := \{ k \mid k
\in C, k \leq_\prefix l, k \not = l\}$.
\end{definition}

The outer clause $\outercl(\prefix, C, l) \subset C$ of $C$ on $l
\in C$ contains only literals that are smaller
than or equal to $l$ in the variable ordering of prefix $\prefix$, excluding $l$.

\begin{definition}[\cite{DBLP:journals/jar/HeuleSB17}]
Let $C$ be a clause with $l \in C$ and $D$ be a clause with $\bar l \in D$
occurring in QBF $\prefix. \clauset$. The \emph{outer resolvent} of $C$ with
$D$ on $l$ with respect to $\prefix$ is the clause $\outerres(\prefix, C, D,
l) := (C \setminus \{l\}) \cup \outercl(\prefix, D, \bar l)$.
\end{definition}

\begin{example}
\label{ex:outerres}
Given $\qclauset := \exists x_1 \forall u \exists x_2. (C \wedge D)$
with $C := (x_1 \vee u \vee x_2)$ and $D := (\bar x_1 \vee \bar u \vee
\bar x_2)$, we have $\outerres(\prefix, C, D, x_1) = ( u \vee x_2) $,
$\outerres(\prefix, C, D, u) = ( x_1 \vee \bar x_1 \vee x_2) $,
$\outerres(\prefix, C, D, x_2) = ( x_1 \vee u \vee \bar x_1
\vee \bar u)$, and
$\outerres(\prefix, D, C, \bar u) = ( x_1 \vee \bar x_1 \vee \bar x_2)$.
Computing
outer resolvents is \mbox{asymmetric since $\outerres(\prefix, C, D, u) \not =
\outerres(\prefix, D, C, \bar u)$.}
\end{example}

\begin{definition}[\cite{DBLP:journals/jar/HeuleSB17}]
\label{def:orig:qior}
Clause $C$ has property \emph{$\qior$ (quantified implied outer
  resolvent)} on literal $l
\in C$ with respect to QBF $\prefix. \clauset$ iff $\clauset \models \outerres(\prefix, C, D, l)$ for all $D
\in \clauset$ with $\bar l \in D$. 
\end{definition}

Property $\qior$ relies on checking whether every possible outer
resolvent $\outerres$ of some clause $C$ on a literal is redundant by checking if $\outerres$ is \emph{propositionally implied} by the \emph{quantifier-free CNF}
$\clauset$ of the given QBF $\prefix. \clauset$. If $C$ has $\qior$ on literal
$l \in C$ then, depending on whether $l$ is existential or universal and
side conditions, either 
$C$ is redundant and can be removed from QBF $\prefix. \clauset$ or
$l$ is redundant and can be removed
from $C$, respectively, resulting in a \emph{satisfiability-equivalent} QBF.

\begin{theorem}[\cite{DBLP:journals/jar/HeuleSB17}]
\label{thm:orig:qior:clause:soundness}
Given a QBF $\qclauset := \prefix. \clauset$ and a clause $C \in \clauset$ with $\qior$ on
an \emph{existential} literal $l \in C$ with respect to QBF
$\qclauset' := \prefix. \clauset'$ where $\clauset' := \clauset \setminus \{C\}$. 
Then $\qclauset \satequiv
\qclauset'$. 
\end{theorem}

\begin{theorem}[\cite{DBLP:journals/jar/HeuleSB17}]
\label{thm:orig:qior:literal:soundness}
Given a QBF $\qclauset_0 := \prefix. \clauset$ and $\qclauset :=
\prefix. (\clauset \cup \{C\})$ where $C$ has $\qior$ on a \emph{universal}
literal $l \in C$ with respect to $\qclauset_0$. Let $\qclauset' :=
\prefix. (\clauset \cup \{C'\})$ with $C' := C \setminus \{l\}$. Then $\qclauset
\satequiv \qclauset'$.
\end{theorem}

Note that in Theorems~\ref{thm:orig:qior:clause:soundness}
and~\ref{thm:orig:qior:literal:soundness} clause $C$ is actually removed
from the QBF for the check whether $C$ has $\qior$ on a literal. 
Checking propositional implication ($\models$) as in Definition~\ref{def:orig:qior} is co-NP
hard and hence intractable. Therefore, in practice a polynomial-time incomplete implication check based on
propositional unit propagation (UP) is applied. The use of UP is central in the $\qrat$ proof system. 

\begin{definition}[propositional unit propagation, UP]
\label{def:prop}
 For a CNF
$\clauset$ and clause $C$, let $\clauset \wedge \overline{C}
\unitprop \emptyset$ denote the fact that \emph{propositional unit propagation
  (UP)} applied to
$\clauset \wedge \overline{C}$ produces the empty
clause, where $\overline{C}$ is the conjunction of the negation of all the literals in
$C$. If $\clauset \wedge \overline{C} \unitprop \emptyset$ then 
we write $\clauset \unitprop C$ to denote that
\emph{$C$ can be derived from
  $\clauset$ by UP} (since $\clauset \models C$).
\end{definition}

\begin{definition}[\cite{DBLP:journals/jar/HeuleSB17}]
\label{def:at}
Clause $C$ has property $\at$ \emph{(asymmetric tautology)} with
respect to a CNF $\clauset$ iff $\clauset \unitprop C$.
\end{definition}

$\at$ is a propositional clause redundancy property that is used in the
$\qrat$ proof system to check whether outer
resolvents are redundant, thereby replacing propositional implication ($\models$) in
Definition~\ref{def:orig:qior} by unit propagation ($\unitprop$) as follows.

\begin{definition}[\cite{DBLP:journals/jar/HeuleSB17}]
\label{def:orig:qrat}
Clause $C$ has property $\qrat$ \emph{(quantified resolution
  asymmetric tautology)} on literal $l \in C$ with respect to QBF
$\prefix. \clauset$ iff, 
for all $D \in \clauset$ with $\bar l \in D$, the outer resolvent 
$\outerres(\prefix, C, D,l)$ has $\at$ with
respect to CNF $\clauset$.
\end{definition}

\begin{example}
Consider $\qclauset := \exists x_1 \forall u \exists x_2. (C \wedge
D)$ with $C := (x_1 \vee u \vee x_2)$ and $D := (\bar x_1 \vee \bar u
\vee \bar x_2)$ from Example~\ref{ex:outerres}. $C$ does not have
$\at$ with respect to CNF $D$, but $C$ has $\qrat$ on $x_2$ with
respect to QBF $\exists x_1 \forall u \exists x_2. (D)$ since
$\outerres(\prefix, C, D, x_2) = ( x_1 \vee u \vee \bar x_1 \vee \bar
u)$ has $\at$ with respect to CNF $D$.
\end{example}

$\qrat$ is a restriction of $\qior$, i.e., a clause that has $\qrat$ also has
$\qior$ but not necessarily vice versa. 
Therefore, the soundness of removing redundant
clauses and literals based on $\qrat$ follows right from
Theorems~\ref{thm:orig:qior:clause:soundness}
and~\ref{thm:orig:qior:literal:soundness}. 

Based on the $\qrat$ redundancy property, the $\qrat$ proof
system~\cite{DBLP:journals/jar/HeuleSB17} consists of rewrite rules to
eliminate redundant clauses, denoted by $\textqrate$,
to add
redundant clauses, denoted by $\textqrata$,
and to
eliminate redundant universal literals, denoted by $\textqratu$.
In a \emph{$\qrat$ satisfaction proof
  (refutation)}, a QBF is reduced to the empty formula
(respectively, to a formula containing the empty clause) by applying
the rewrite rules. The $\qrat$ proof systems has an additional rule to
eliminate universal literals by \emph{extended universal reduction}
($\texteur$).  We do not present $\texteur$ because it is not
affected by our generalization of $\qrat$, which we define in the
following. Observe that $\qior$ and $\at$ (and hence also $\qrat$)
are based on \emph{propositional} implication
($\models$) and unit propagation ($\unitprop$), i.e., the
quantifier structure of the given QBF is not exploited.


\section{$\newqrat$: A More Powerful QBF Redundancy Property} \label{sec:new:qrat}

We make redundancy checking of outer resolvents in $\qrat$ aware of the
quantifier structure of a QBF. To this end, we generalize $\qior$ and $\at$ by
replacing propositional implication ($\models$) and unit propagation
($\unitprop$) by QBF implication $(\qbfmodels)$ and QBF unit propagation,
respectively. Thereby, we obtain a more general and more powerful notion of
the $\qrat$ redundancy property, which we call $\newqrat$.  

First, in
Proposition~\ref{prop:qior:qbf:model:equiv} we point out a property of $\qior$
(Definition~\ref{def:orig:qior}) which is due to the following result from
related work~\cite{DBLP:conf/cp/SamulowitzDB06}: if we attach a quantifier
prefix $\prefix$ to equivalent CNFs $\clauset$ and $\clauset'$, then the
resulting QBFs are equivalent.

\begin{proposition}[\cite{DBLP:conf/cp/SamulowitzDB06}]
\label{prop:sat:equiv:implies:qbf:equiv}
Given CNFs $\clauset$ and $\clauset'$ such that $\vars(\clauset) =
\vars(\clauset')$ and a quantifier prefix $\prefix$ defined precisely over 
$\vars(\clauset)$. If $\clauset \equiv \clauset'$ then $\prefix. \clauset \treemodelequiv \prefix. \clauset'$.
\end{proposition}

\begin{proposition}
\label{prop:qior:qbf:model:equiv}
If clause $C$ has $\qior$ on literal $l
\in C$ with respect to QBF $\prefix. \clauset$, then $\prefix. \clauset
\treemodelequiv \prefix. (\clauset \wedge \outerres(\prefix, C, D, l))$  for all $D
\in \clauset$ with $\bar l \in D$.
\end{proposition}
\begin{proof}
Since $C$ has $\qior$ on literal $l \in C$ with respect to QBF
$\prefix. \clauset$, by Definition~\ref{def:orig:qior} we have $\clauset
\models \outerres(\prefix, C, D, l)$ for all $D \in \clauset$ with $\bar l \in
D$, and further also $\clauset
\equiv \clauset \wedge \outerres(\prefix, C, D, l)$.  Then $\prefix. \clauset
\treemodelequiv \prefix. (\clauset \wedge \outerres(\prefix, C, D, l))$ by
Proposition~\ref{prop:sat:equiv:implies:qbf:equiv}.
\qed
\end{proof}

By Proposition~\ref{prop:qior:qbf:model:equiv} any outer resolvent $\outerres$
of some clause $C$ that has $\qior$ with respect to some QBF
$\prefix. \clauset$ is redundant in the sense that it can be \emph{added to
the QBF $\prefix. \clauset$ in an equivalence preserving way}
($\treemodelequiv$), i.e., $\outerres$ is \emph{implied by the QBF}
$\prefix. \clauset$ ($\qbfmodels$). This is the central characteristic of our
generalization $\newqrat$ of $\qrat$. We develop a redundancy property used in
$\newqrat$ which allows to, e.g., remove a clause $C$ from a QBF
$\prefix. \clauset$ in a satisfiability preserving way (like in $\qrat$,
cf.~Theorem~\ref{thm:orig:qior:clause:soundness}.) if all respective outer
resolvents of $C$ are implied by the QBF $\prefix. (\clauset \setminus \{C\})$.  Since checking QBF
implication is intractable just like checking propositional implication in
$\qior$, in practice we apply a polynomial-time incomplete QBF implication
check based on \emph{QBF unit propagation}.

In the following, we develop a theoretical framework of \emph{abstractions} of 
QBFs that underlies our generalization $\newqrat$ of $\qrat$. Abstractions are
crucial for the soundness of checking QBF implication by QBF unit propagation.

\begin{definition}[nesting levels, prefix/QBF abstraction] 
\label{def:abs}
Let $\qclauset := \prefix. \clauset$ be a QBF with prefix $\prefix := Q_1B_1
\ldots Q_iB_i Q_{i+1}B_{i+1} \ldots Q_nB_n$. 
For a clause $C$, $\levels(\prefix, C) := \{i \mid
\exists l \in C, \quant{\prefix}{l} = Q_i\}$ is the set of
\emph{nesting levels} in $C$.\footnote{In general, clauses $C$ 
  are always (implicitly) interpreted under a quantifier prefix $\prefix$.} 
The \emph{abstraction of $\prefix$} with
respect to $i$ with $0 \leq i \leq n$ produces the \emph{abstracted prefix}
$\eabs(\prefix, i) := \prefix$ for $i = 0$ and otherwise $\eabs(\prefix, i) :=
\exists (B_1 \cup \ldots \cup B_i) Q_{i+1}B_{i+1} \ldots Q_nB_n$.
The \emph{abstraction of $\qclauset$} with
respect to $i$ with $0 \leq i \leq n$ produces the \emph{abstracted QBF} $\eabs(\qclauset, i) :=
\eabs(\prefix, i). \clauset$ with prefix $\eabs(\prefix, i)$.
\end{definition}

\begin{example}
\label{ex:basic:eabs}
Given the QBF $\qclauset := \prefix. \clauset$ with prefix 
$\prefix := \forall B_1 \exists B_2 \forall B_3 \exists B_4$. We have $\eabs (\qclauset, 0) =
\qclauset$, $\eabs (\qclauset, 1) = \eabs (\qclauset, 2) = \exists (B_1 \cup
B_2) \forall B_3 \exists B_4. \clauset$, $\eabs
(\qclauset, 3) = \eabs (\qclauset, 4) = \exists (B_1 \cup B_2 \cup B_3
\cup B_4). \clauset$. 
\end{example}

In an abstracted QBF $\eabs(\qclauset, i)$ universal variables from
blocks smaller than or equal to $B_i$ are converted into existential
ones. If the original QBF $\qclauset$ has a model $T$, then \emph{all}
nodes in $T$ associated to universal variables must be labelled with
$\top$, in particular the universal variables that are existential in
$\eabs(\qclauset, i)$. Hence, for \emph{all} models $T$ of $\qclauset$, \emph{every}
model $T^A$ of $\eabs(\qclauset, i)$ is a subtree of $T$.

\begin{proposition}
\label{prop:abs:property:new}
Given a QBF $\qclauset := \prefix. \clauset$ with prefix $\prefix := Q_1B_1
\ldots Q_iB_i \ldots Q_nB_n$ and $\eabs(\qclauset,
i)$ for some arbitrary $i$ with $0 \leq i \leq n$.
For all $T$ and $T^A$ we have that if $T \qbfmodels
\qclauset$ and $T^A \subseteq T$ is a pre-model of $\eabs(\qclauset,i)$, then $T^{A}
\qbfmodels \eabs(\qclauset,i)$. 
\end{proposition}
\begin{proof}
  By induction on $i$.  The base case $i := 0$ is trivial.

As induction hypothesis (IH), assume that the claim holds for some $i$ with $0
\leq i < n$, i.e., for all $T$ and $T^A$ we have that if $T \qbfmodels \qclauset$ and $T^A \subseteq T$
is a pre-model of $\eabs(\qclauset,i)$, then $T^{A} \qbfmodels
\eabs(\qclauset,i)$.  Consider $\eabs(\qclauset, j)$ for $j = i+1$,
which is an abstraction of $\eabs(\qclauset, i)$.  We have to show
that, for all $T$ and $T^B$ we have that if $T \qbfmodels \qclauset$ and $T^B \subseteq T$ is a pre-model
of $\eabs(\qclauset,j)$, then $T^{B} \qbfmodels \eabs(\qclauset,j)$.
We distinguish cases by the type of $Q_j$ in the abstracted prefix
$\eabs(\prefix, i) = \exists (B_1 \cup \ldots \cup B_i) Q_jB_j \ldots
Q_nB_n$ of $\eabs(\qclauset, i)$.

If $Q_j = \exists$ then $\eabs(\prefix, i) = \eabs(\prefix, j) = \exists (B_1 \cup \ldots B_i \cup B_j) \ldots
Q_nB_n$. Since $\eabs(\qclauset, i) = \eabs(\qclauset, j)$, the claim holds for $\eabs(\qclauset, j)$
by IH.

If $Q_j = \forall$ then, towards a contradiction, assume that, for some $T$ and $T^B$, $T
\qbfmodels \qclauset$ and $T^B \subseteq T$ is a pre-model of
$\eabs(\qclauset,j)$, but $T^{B} \not \qbfmodels
\eabs(\qclauset,j)$. Then the root of $T^B$ is labelled with $\bot$,
and in particular the nodes of all the variables which are existential
in $B_j$ with respect to $\eabs(\prefix, j)$ are also labelled with
$\bot$. These existential variables appear along a single branch
$\tau'$ in $T^B$, i.e., $\tau'$ is a partial assignment of the
variables in $B_j$. Since $T^B \subseteq T^A$ and $Q_j = \forall$ in
$\eabs(\prefix, i)$, the root of $T^A$ is labelled with $\bot$ since
there is the branch $\tau'$ containing the variables in $B_j$ whose
nodes are labelled with $\bot$ in $T^A$. Hence $T^A \not \qbfmodels
\eabs(\qclauset,i)$, which is a contradiction to IH. Therefore, we
conclude that $T^{B} \qbfmodels \eabs(\qclauset,j)$.
\qed
\end{proof}

If an abstraction $\eabs(\qclauset, i)$ is unsatisfiable then also the original
QBF $\qclauset$ is
unsatisfiable due to Proposition~\ref{prop:abs:property:new}. We generalize
Proposition~\ref{prop:sat:equiv:implies:qbf:equiv} from CNFs to QBFs
and their abstractions. Note that the full abstraction $\eabs(\qclauset, i)$ for
$i := n$ of a QBF $\qclauset$ is a CNF, i.e., it does not contain any universal variables.

\begin{lemma}
\label{lem:abs:equiv:implies:qbf:equiv}
Let $\qclauset := \prefix. \clauset$ and $\qclauset' := \prefix. \clauset'$
be QBFs with the same prefix $\prefix := Q_1B_1 \ldots Q_iB_i \ldots Q_nB_n$. Then for all $i$, if $\eabs(\qclauset, i) 
\treemodelequiv \eabs(\qclauset', i)$ then 
$\qclauset \treemodelequiv \qclauset'$.

\end{lemma}
\begin{proof}
By induction on $i := 0$ up to $i := n$.  The base case $i := 0$ is trivial.

As induction hypothesis (IH), assume that the claim holds for some $i$ with $0
\leq i < n$, i.e., if $\eabs(\qclauset, i) \treemodelequiv
\eabs(\qclauset', i)$ then $\qclauset \treemodelequiv \qclauset'$.
Let $j = i+1$ and consider $\eabs(\qclauset, j)$ and
$\eabs(\qclauset', j)$, which are abstractions of $\eabs(\qclauset,
i)$ and $\eabs(\qclauset', i)$.  We have $\eabs(\prefix, i) = \exists
(B_1 \cup \ldots \cup B_i) Q_jB_j \ldots Q_nB_n$ and $\eabs(\prefix,
j) = \exists (B_1 \cup \ldots \cup B_j) \ldots Q_nB_n$.  We show that
if $\eabs(\qclauset, j) \treemodelequiv \eabs(\qclauset', j)$ then
$\eabs(\qclauset, i) \treemodelequiv \eabs(\qclauset', i)$, and hence
also $\qclauset \treemodelequiv \qclauset'$ by IH.  Assume that
$\eabs(\qclauset, j) \treemodelequiv \eabs(\qclauset', j)$. We
distinguish cases by the type of $Q_j$ in $\eabs(\prefix, i)$.  If
$Q_j = \exists$ then $\eabs(\prefix, i) = \eabs(\prefix, j) = \exists
(B_1 \cup \ldots B_i \cup B_j) \ldots Q_nB_n$, and hence
$\eabs(\qclauset, i) \treemodelequiv \eabs(\qclauset', i)$.

If $Q_j = \forall$, then towards a contradiction, assume that
$\eabs(\qclauset, j) \treemodelequiv \eabs(\qclauset', j)$ but
$\eabs(\qclauset, i) \not \treemodelequiv \eabs(\qclauset', i)$.  Then
there exists $T$ such that $T \qbfmodels \eabs(\qclauset, i)$ but $T
\not \qbfmodels \eabs(\qclauset',i)$.  Since $T \not \qbfmodels
\eabs(\qclauset',i)$ there exists a pre-model $T^{A} \subseteq T$ of
$\eabs(\qclauset',j)$ such that the root of $T^A$ is labelled with
$\bot$, and in particular the nodes of all the variables which are
existential in $B_j$ with respect to $\eabs(\prefix, j)$ (and
universal with respect to $\eabs(\prefix, i)$) are also labelled with
$\bot$. These existential variables appear along a single branch
$\tau'$ in $T^A$, i.e., $\tau'$ is a partial assignment of the
variables in $B_j$.  Therefore we have $T^A \not \qbfmodels
\eabs(\qclauset', j)$.  Since $T \qbfmodels \eabs(\qclauset,i)$ and
$T^A \subseteq T$, we have
$T^A \qbfmodels \eabs(\qclauset, j)$ by Proposition~\ref{prop:abs:property:new}, which contradicts the assumption
that $\eabs(\qclauset, j) \treemodelequiv \eabs(\qclauset', j)$.
\qed 
\end{proof}

The converse of Lemma~\ref{lem:abs:equiv:implies:qbf:equiv} 
does not hold. From the
equivalence of two QBFs $\qclauset$ and $\qclauset'$  we
cannot conclude that the abstractions $\eabs(\qclauset, i)$ and $\eabs(\qclauset', i)$ are
equivalent. In our
generalization $\newqrat$ of the $\qrat$ system we  
check whether an outer
resolvent of some clause $C$ 
is implied ($\qbfmodels$) by an \emph{abstraction} of the given QBF. If so
then by Lemma~\ref{lem:abs:equiv:implies:qbf:equiv} the outer resolvent is
also implied by the original QBF. Below we prove that this condition is
sufficient for the soundness of redundancy removal in $\newqrat$. To check QBF
implication in an incomplete way and in polynomial time, in practice we apply \emph{QBF
unit propagation},  which is an extension of propositional unit propagation, to
abstractions of the given QBF. 

\begin{definition}[universal reduction, UR~\cite{DBLP:journals/iandc/BuningKF95}]
\label{def:forallred}
Given a QBF $\qclauset := \prefix. \clauset$ and a non-tautological clause $C$, \emph{universal
  reduction (UR)} of $C$ produces the clause $\UR(\prefix, C) :=
C \setminus \{l \in C \mid \quant{\prefix}{l} = \forall, \forall l' \in
C, \quant{\prefix}{l'} = \exists: \var{l'} \leq_\prefix \var{l}\}$.
\end{definition}

\begin{definition}[QBF unit propagation, QUP]
\label{def:qprop}
\emph{QBF unit propagation (QUP)} extends UP (Definition~\ref{def:prop}) by 
applications of UR. For a QBF
$\qclauset := \prefix. \clauset$ and a clause $C$, let 
$\prefix. (\clauset \wedge \overline{C}) \unitQprop \emptyset$ 
denote the fact that QUP applied to
$\prefix. (\clauset \wedge \overline{C})$ produces the empty
clause, where $\overline{C}$ is the conjunction of the negation of all the literals in
$C$. If $\prefix. (\clauset \wedge \overline{C}) \unitQprop \emptyset$ and
additionally $\prefix. \clauset \qbfmodels \prefix. (\clauset \wedge C)$ then we
write $\qclauset \unitQprop C$ to denote that
\emph{$C$ can be derived from
  $\qclauset$ by QUP}.
\end{definition}

In contrast to UP (Definition~\ref{def:prop}), 
deriving the empty clause by QUP by propagating $\overline{C}$ on a QBF $\qclauset$ 
 is not sufficient to conclude that $C$ is implied by $\qclauset$. 

\begin{example}
\label{ex:qprop}
Given the QBF $\prefix. \clauset$ with prefix $\prefix :=
\forall u \exists x$ and CNF 
$\clauset := (u \vee \bar x) \wedge (\bar u \vee x)$ 
and the clause $C := (x)$. We have 
$\prefix. ((u \vee \bar x) \wedge (\bar u \vee x) \wedge (\bar x)) \unitQprop
\emptyset$ since propagating $\overline{C} = (\bar x)$ produces 
$(\bar u)$, which is reduced to $\emptyset$ by UR. However, 
$\prefix. \clauset \not \qbfmodels \prefix. (\clauset \wedge C)$ since $\prefix. \clauset$ is satisfiable
whereas $\prefix. (\clauset \wedge C)$ is unsatisfiable.  Note that $\eabs(\prefix. ((u \vee \bar x) \wedge (\bar u \vee x) \wedge
\bar x), 2) \nunitQprop \emptyset$.
\end{example}

To correctly apply QUP for checking whether some clause $C$ (e.g., an
outer resolvent) is implied by a QBF $\qclauset := \prefix. \clauset$ and thus avoid the problem
illustrated in Example~\ref{ex:qprop}, we carry out QUP on a
\emph{suitable abstraction} of $\qclauset$ with respect to $C$.  Let
$i = \max(\levels(\prefix, C))$ be the maximum nesting level of
variables that appear in $C$. We show that if QUP derives the empty
clause from the abstraction $\eabs(\qclauset, i)$ augmented by the
negated clause $\overline{C}$, i.e., $\eabs(\prefix. (\clauset \wedge
\overline{C}), i) \unitQprop \emptyset$, then we can safely conclude
that $C$ is implied by the \emph{original} QBF, i.e.,
$\prefix. \clauset \qbfmodels \prefix. (\clauset \wedge C)$. 
This approach extends failed literal detection for QBF 
preprocessing~\cite{DBLP:conf/sat/LonsingB11}.

\begin{lemma}
\label{lem:outer:qprop:soundness}
Let $\prefix. \clauset$ be a QBF with prefix $\prefix := Q_1B_1
\ldots Q_nB_n$ and $C$ a clause such that $\vars(C) \subseteq B_1$. If $\prefix. (\clauset \wedge \overline{C}) \unitQprop
\emptyset$ then $\prefix. \clauset \treemodelequiv \prefix. (\clauset \wedge C)$.
\end{lemma}
\begin{proof}
By contradiction, assume $T \qbfmodels \prefix. \clauset$ but 
$T \not \qbfmodels \prefix. (\clauset \wedge C)$.  Then there
is a path $\tau \subseteq T$ such that $\tau(C) = \bot$.  Since $\vars(C) \subseteq B_1$ and $\prefix. (\clauset
\wedge \overline{C}) \unitQprop \emptyset$, the QBF $\prefix. (\clauset \wedge
\overline{C})$ is unsatisfiable and in particular $T \not \qbfmodels \prefix. (\clauset \wedge
\overline{C})$. Since $\tau(C) = \bot$, we have 
$\tau(\overline{C}) = \top$ and hence $T \qbfmodels \prefix. (\clauset
\wedge \overline{C})$, which is a contradiction.
\qed
\end{proof}

\begin{lemma}
\label{lem:abs:qprop:soundness}
Let $\prefix. \clauset$ be a QBF, $C$ a clause,
and $i = \max(\levels(\prefix, C))$. 
If $ \eabs(\prefix. (\clauset \wedge \overline{C}), i) \unitQprop
\emptyset$ then $\eabs(\prefix. \clauset, i) \treemodelequiv
\eabs(\prefix. (\clauset \wedge C), i)$.
\end{lemma}
\begin{proof}
The claim follows from Lemma~\ref{lem:outer:qprop:soundness} since all
variables that appear in $C$ are existentially quantified in
$\eabs(\prefix. (\clauset \wedge \overline{C}), i)$ in the leftmost quantifier
block.
\qed
\end{proof}

\begin{lemma}
\label{lem:abs:treemodel:soundness}
Let $\prefix. \clauset$ be a QBF, $C$ a clause,
and $i = \max(\levels(\prefix, C))$. 
If $ \eabs(\prefix. (\clauset \wedge \overline{C}), i) \unitQprop
\emptyset$ then $\prefix. \clauset \treemodelequiv 
\prefix. (\clauset \wedge C)$.
\end{lemma}
\begin{proof}
By Lemma~\ref{lem:abs:qprop:soundness} and Lemma~\ref{lem:abs:equiv:implies:qbf:equiv}.
\qed
\end{proof}

Lemma~\ref{lem:abs:treemodel:soundness} provides us with the necessary
theoretical foundation to lift $\at$ (Definition~\ref{def:at}) from
UP, which is applied to CNFs, to QUP, which is applied to \emph{suitable
abstractions} of QBFs. The abstractions are constructed depending on
the maximum nesting level of variables in the clause we want to
check. 

\begin{definition}[$\qat$]
\label{def:qat} Let $\qclauset$ be a QBF, $C$ a clause,
and $i = \max(\levels(\prefix, C))$ 
Clause $C$ has property $\qat$ \emph{(quantified asymmetric
tautology)} with respect to $\qclauset$ iff $\eabs(\qclauset, i) \unitQprop
C$.
\end{definition}

As an immediate consequence from the definition of QUP
(Definition~\ref{def:qprop}) and Lemma~\ref{lem:abs:qprop:soundness},
we can conclude that a clause $C$ has $\qat$ with respect to a QBF
$\prefix. \clauset$ if QUP derives the empty clause from the suitable
abstraction of $\prefix. \clauset$ with respect to $C$ (i.e.,
$\eabs(\prefix. (\clauset \wedge \overline{C}), i) \unitQprop
\emptyset$). Further, if $C$ has $\qat$ then we have
$\prefix. \clauset \treemodelequiv \prefix. (\clauset \wedge C)$ by
Lemma~\ref{lem:abs:treemodel:soundness}, i.e., $C$ is implied by
the given QBF $\prefix. \clauset$.

\begin{example}
\label{ex:at:qat}
Given the QBF $\qclauset := \prefix. \clauset$ with $\prefix :=
\forall u_1 \exists x_3 \forall u_2 \exists x_4$ and $\clauset := (u_1
\vee \bar x_3) \wedge (u_1 \vee \bar x_3 \vee x_4) \wedge (\bar u_2
\vee \bar x_4)$.  Clause $(u_1 \vee \bar x_3)$ has $\qat$ with respect
to $\eabs(\qclauset, 2)$ with $\max(\levels(C)) = 2$ since $\forall
u_2$ is still universal in the abstraction. By QUP clause $(u_1 \vee
\bar x_3 \vee x_4)$ becomes unit and clause $(\bar u_2 \vee \bar x_4)$
becomes empty by UR. However, clause $(u_1 \vee \bar x_3)$ does
\emph{not} have $\at$ since $\forall u_2$ is treated as an existential
variable in UP, hence clause $(\bar u_2 \vee \bar x_4)$ does not
become empty by UR.
\end{example}

In contrast to $\at$, $\qat$ is aware of quantifier structures in
QBFs as shown in Example~\ref{ex:at:qat}. We now generalize
$\qrat$ to $\newqrat$ by replacing $\at$ by $\qat$. Similarly,
we generalize $\qior$ to $\newqior$ by replacing propositional
implication ($\models$) and equivalence
(Proposition~\ref{prop:sat:equiv:implies:qbf:equiv}), by \emph{QBF
implication and equivalence} (Lemma~\ref{lem:abs:treemodel:soundness}).

\begin{definition}[$\newqrat$]
\label{def:new:qrat}
Clause $C$ has property $\newqrat$  on literal $l \in C$ with respect to QBF
$\prefix. \clauset$ iff, 
for all $D \in \clauset$ with $\bar l \in D$, the outer resolvent 
$\outerres(\prefix, C, D,l)$ has $\qat$ with
respect to QBF $\prefix. \clauset$.
\end{definition}

\begin{definition}[$\newqior$]
\label{def:new:qior}
Clause $C$ has property $\newqior$  on literal $l
\in C$ with respect to QBF $\prefix. \clauset$ iff $\prefix. \clauset \treemodelequiv \prefix. (\clauset \wedge \outerres(\prefix, C, D, l))$ for all $D
\in \clauset$ with $\bar l \in D$. 
\end{definition}

If a clause has $\qrat$ then it also has $\newqrat$. Moreover, due to
Proposition~\ref{prop:qior:qbf:model:equiv}, if a clause has $\qior$ then it
also has $\newqior$. Hence $\newqrat$ and $\newqior$ indeed are
generalizations of $\qrat$ and $\qior$, which are strict, as we argue below.
The soundness of removing redundant clauses and universal literals based on
$\newqior$ (and on $\newqrat$) can be proved by the \emph{same}
arguments as original $\qrat$, which we outline in the
following. We refer to the appendix for full proofs.

\begin{definition}[prefix/suffix assignment~\cite{DBLP:journals/jar/HeuleSB17}]
\label{def:prefix:suffix}
For a QBF $\qclauset := \prefix. \clauset$ and a complete assignment $\tau$
in the assignment tree of $\qclauset$, the partial \emph{prefix}
and \emph{suffix assignments} of $\tau$ with respect to variable $x$,
denoted by $\tau^x$ and $\tau_x$, respectively, are defined as 
$\tau^x := \{y \mapsto \tau(y) \mid y \leq_\prefix x, y \not = x\}$ and
$\tau_x := \{y \mapsto \tau(y) \mid y \not \leq_\prefix x\}$.
\end{definition}

For a variable $x$ from block $B_i$ of a QBF, Definition~\ref{def:prefix:suffix} allows us to split a
complete assignment $\tau$ into three parts $\tau^xl\tau_x$, where the prefix assignment $\tau^x$ 
assigns variables (excluding $x$) from blocks smaller than or equal
to $B_i$, $l$ is a literal of $x$, and the suffix
assignment $\tau_x$ assigns variables from \mbox{blocks larger than $B_i$.}

Prefix and suffix assignments are important for proving the soundness of
satisfiability-preserving redundancy removal by $\newqior$ (and
$\qior$). Soundness is proved by showing that certain paths in a model
of a QBF can safely be modified based on prefix and suffix assignments, as stated in the following.

\begin{lemma}[cf.~Lemma~6 in~\cite{DBLP:journals/jar/HeuleSB17}]
\label{lem:outerres:satisfied}
Given a clause $C$ with $\newqior$  with respect to QBF
$\qclauset := \prefix. \clauset$ on literal $l \in C$ with $\var{l} = x$. Let 
$T$ be a model of $\qclauset$ and  $\tau \subseteq T$ be a path in $T$. If
$\tau(C \setminus \{l\}) =
\bot$ then $\tau^x(D) = \top$ for all $D \in \clauset$ \mbox{with $\bar l \in D$.}
\end{lemma}
\begin{proof}[sketch, see appendix]
Let $D \in \clauset$ be a clause with $\bar l \in D$ and $R :=
\outerres(\prefix,C,D,l) = (C \setminus \{l\}) \cup \outercl(\prefix, D, \bar
l)$.  By Definition~\ref{def:new:qior}, we have $\prefix. \clauset
\treemodelequiv \prefix. (\clauset \wedge \outerres(\prefix, C, D, l))$ for
all $D \in \clauset$ with $\bar l \in D$.  The rest of the proof considers a
path $\tau$ in $T$ and works in the same way as 
the proof of Lemma~6 in~\cite{DBLP:journals/jar/HeuleSB17}.
\qed
\end{proof}

\begin{theorem}
\label{thm:newqior:clause:soundness}
Given a QBF $\qclauset := \prefix. \clauset$ and a clause $C \in \clauset$ with $\newqior$ on
an \emph{existential} literal $l \in C$ with respect to QBF
$\qclauset' := \prefix. \clauset'$ where $\clauset' := \clauset \setminus \{C\}$. 
Then $\qclauset \satequiv
\qclauset'$. 
\end{theorem}
\begin{proof}[sketch, see appendix] The proof relies on
Lemma~\ref{lem:outerres:satisfied} and works in the same way as the proof of
Theorem~7 in~\cite{DBLP:journals/jar/HeuleSB17}. A model $T$ of $\qclauset$
is obtained from a model $T'$ of $\qclauset'$ by flipping the assignment
of variable $x = \var{l}$ on a path $\tau$ in $T'$ to satisfy clause $C$. All $D \in
\clauset$ with $\bar l \in D$ are satisfied by such modified $\tau$.
\qed
\end{proof}

\begin{theorem}
\label{thm:newqior:literal:soundness}
Given a QBF $\qclauset_0 := \prefix. \clauset$ and $\qclauset :=
\prefix. (\clauset \cup \{C\})$ where $C$ has $\newqior$ on a \emph{universal}
literal $l \in C$ with respect to $\qclauset_0$. Let $\qclauset' :=
\prefix. (\clauset \cup \{C'\})$ with $C' := C \setminus \{l\}$. Then $\qclauset
\satequiv \qclauset'$.
\end{theorem}
\begin{proof}[sketch, see appendix]
The proof relies on
Lemma~\ref{lem:outerres:satisfied} and works in the same way as the proof of Theorem~8
in~\cite{DBLP:journals/jar/HeuleSB17}. A model $T'$ of $\qclauset'$
is obtained from a model $T$ of $\qclauset$ by modifying the subtree under the
node associated to variable $x = \var{l}$. Suffix assignments of some paths
$\tau$ in $T$ are used to construct modified paths in $T'$ under which clause $C'$ is
satisfied. All $D \in
\clauset$ with $\bar l \in D$ are still satisfied after such modifications.
\qed
\end{proof}

Analogously to the $\qrat$ proof system that is based on the $\qrat$ redundancy
property (Definition~\ref{def:orig:qrat}), we obtain the \emph{$\newqrat$ proof
system} based on property $\newqrat$ (Definition~\ref{def:new:qrat}). The
system consists of rewrite rules $\textnewqrate$, $\textnewqrata$, and
$\textnewqratu$ to eliminate or add redundant clauses, and to eliminate redundant
universal literals. On a conceptual level, these rules in $\newqrat$ are similar to their respective
counterparts in the $\qrat$ system. The extended universal reduction rule
$\texteur$ is the same in the $\qrat$ and $\newqrat$ systems. In contrast to
$\qrat$, $\newqrat$ is aware of  quantifier structures of QBFs because it
relies on the QBF specific property $\qat$ and QUP instead of on propositional $\at$ and UP.

The $\newqrat$ system has the same desirable properties as the original
$\qrat$ system. $\newqrat$ \emph{simulates} virtually all inference rules
applied in QBF reasoning tools and it is based on redundancy
property $\newqrat$ that can be checked in \emph{polynomial time} by QUP. 
Further, $\newqrat$ allows to represent proofs in the
\emph{same proof format} as
$\qrat$. However, proof checking, i.e., checking whether a clause listed in
the proof has $\newqrat$ on a literal, must be adapted to the use of
QBF abstractions and QUP. Consequently, the available $\qrat$ proof
checker
\textsf{QRATtrim}~\cite{DBLP:journals/jar/HeuleSB17} 
cannot be used out of the box to check $\newqrat$ proofs.

Notably, \emph{Skolem functions} can be extracted from $\newqrat$
proofs of satisfiable QBFs in the \emph{same} way as in $\qrat$
(consequence of Theorem~\ref{thm:newqior:clause:soundness},
cf.~Corollaries~26 and~27
in~\cite{DBLP:journals/jar/HeuleSB17}). 
Hence like
$\qrat$, $\newqrat$ can be
integrated in complete QBF workflows that include preprocessing,
solving, and Skolem function extraction~\cite{DBLP:conf/tap/FazekasHSB17}.


\section{Exemplifying the Power of $\newqrat$} \label{sec:example}

In the following, we point out that the $\newqrat$ system is more powerful
than $\qrat$ in terms of redundancy detection. In particular, we show that
the rules $\textnewqrate$ and $\textnewqratu$ in the
$\newqrat$ system can eliminate certain redundancies that their counterparts
$\textqrate$ and $\textqratu$ cannot eliminate.

\begin{definition}\label{def:formula:class}

For $n \geq 1$, let $\FormulasC(n) := \prefixC(n). \clausetC(n)$ be a class of QBFs with prefix
$\prefixC(n)$ and CNF $\clausetC(n)$ defined as follows.

\medskip

\hspace*{-0.5cm}
\begin{minipage}[t]{0.425\textwidth}
$\prefixC(n) := \exists B_1 \forall B_2 \exists B_3 \forall B_4 \exists B_5$:

\begin{tabular}{ll}
$B_1 := \{ x_{4i + 1}, x_{4i + 2} \mid 0 \leq i < n\}$ & \\
$B_2 := \{ u_{2i + 1} \mid 0 \leq i < n \}$ & \\
$B_3 := \{ x_{4i + 3} \mid 0 \leq i < n\}$ & \\
$B_4 := \{ u_{2i + 2} \mid 0 \leq i < n \}$ & \\
$B_5 := \{ x_{4i + 4} \mid 0 \leq i < n \}$ & \\
\end{tabular}
\end{minipage}
\begin{minipage}[t]{0.55\textwidth}
$\clausetC(n) := \bigwedge_{i := 0}^{n-1} \mathcal{C}(i)$ with $\mathcal{C}(i)
:= \bigwedge_{j := 0}^{6} C_{i,j}$:

\begin{tabular}{ll}
$C_{i,0} := $ & $(x_{4i + 1} \vee u_{2i + 1} \vee \neg x_{4i + 3} )$ \\
$C_{i,1} := $ & $(x_{4i + 2} \vee \neg u_{2i + 1} \vee  x_{4i + 3}) $ \\
$C_{i,2} := $ & $(\neg x_{4i + 1} \vee \neg u_{2i + 1} \vee \neg x_{4i + 3}) $ \\
$C_{i,3} := $ & $(\neg x_{4i + 2} \vee u_{2i + 1} \vee  x_{4i + 3})  $  \\
$C_{i,4} := $ & $(u_{2i + 1} \vee \neg x_{4i + 3} \vee x_{4i + 4}) $  \\
$C_{i,5} := $ & $(\neg u_{2i + 2} \vee \neg x_{4i + 4}) $ \\
$C_{i,6} := $ & $(\neg x_{4i + 1} \vee u_{2i + 2} \vee \neg x_{4i + 4})$ \\
\end{tabular}
\end{minipage}

\end{definition}

\begin{example} \label{ex:formula:class:length:one}
For $n := 1$, we have $\FormulasC(n)$ with prefix $\prefixC(n) := 
\exists  x_{1}, x_{2} 
\forall  u_{1} 
\exists  x_{3} 
\forall  u_{2} 
\exists  x_{4} 
$
and CNF $\clausetC(n) := \mathcal{C}(0)$ with $\mathcal{C}(0)
:= \bigwedge_{j := 0}^{6} C_{0,j}$ as follows.

\begin{minipage}[t]{0.5\textwidth}
\hphantom{dummy}

\begin{tabular}{ll}
$C_{0,0} := $ & $(x_{1} \vee u_{1} \vee \neg x_{3} )$ \\
$C_{0,1} := $ & $(x_{2} \vee \neg u_{1} \vee  x_{3})$ \\
$C_{0,2} := $ & $(\neg x_{1} \vee \neg u_{1} \vee \neg x_{3})$ \\
$C_{0,3} := $ & $(\neg x_{2} \vee u_{1} \vee  x_{3})$  \\
\end{tabular}
\end{minipage}
\begin{minipage}[t]{0.5\textwidth}
\hphantom{dummy}

\begin{tabular}{ll}
$C_{0,4} := $ & $(u_{1} \vee \neg x_{3} \vee x_{4})$  \\
$C_{0,5} := $ & $(\neg u_{2} \vee \neg x_{4})$ \\
$C_{0,6} := $ & $(\neg x_{1} \vee u_{2} \vee \neg x_{4})$ \\
\end{tabular}
\end{minipage}

\end{example}

\begin{proposition} \label{prop:stronger:qrat}
For $n \geq 1$, $\textnewqrate$ can eliminate all clauses in $\FormulasC(n)$
whereas $\textqrate$ cannot eliminate any clause in $\FormulasC(n)$.
\end{proposition}
\begin{proof}[sketch]
For $i$ and $k$ with $i \not = k$, the sets of
variables in $\mathcal{C}(i)$ and $\mathcal{C}(k)$ are disjoint.  Thus
it suffices to prove the claim for an arbitrary $\mathcal{C}(i)$. 
Clause $C_{i,0}$ has $\newqrat$ on literal $x_{4i + 1}$ and can be
removed. The relevant outer resolvents are $\outerres_{0,2} =
\outerres(\prefixC(n), C_{i,0}, C_{i,2}, x_{4i + 1})$ and
$\outerres_{0,6} = \outerres(\prefixC(n), C_{i,0}, C_{i,6}, x_{4i +
  1})$, and we have $\outerres_{0,2} = \outerres_{0,6} = (u_{2i + 1}
\vee \neg x_{4i + 3})$.  Since $\max(\levels(\outerres_{0,2})) =
\max(\levels(\outerres_{0,6})) = 3$, we apply QUP to the abstraction
$\eabs(\FormulasC(n), 3)$. Note that variable $u_{2i + 2}$ from block
$B_4$ still is universal in the prefix of $\eabs(\FormulasC(n), 3)$.
Propagating $\overline{\outerres_{0,2}}$ and
$\overline{\outerres_{0,6}}$, respectively, in either case makes 
$C_{i,4}$ unit, finally $C_{i,5}$ becomes empty under the
derived assignment $x_{4i + 4}$ since UR reduces the literal $\neg
u_{2i + 2}$. After removing $C_{i,0}$, 
clauses $C_{i,2}$ and $C_{i,6}$ trivially have $\newqrat$ on $\neg
x_{4i + 1}$.  Then $C_{i,1}$ has $\newqrat$
on $x_{4i + 3}$. Finally, the remaining clauses trivially have
$\newqrat$.
In contrast to that, $\textqrate$ cannot eliminate any clause in
$\FormulasC(n)$. Clause $C_{i,5}$ does not become empty by UP since all variables are existential. The claim can
be proved by case analysis of all possible outer
resolvents.  
\qed
\end{proof}


\begin{definition}\label{def:formula:class:literal:elimination}

For $n \geq 1$, let $\FormulasL(n) := \prefixL(n). \clausetL(n)$ be a class of QBFs with prefix
$\prefixL(n)$ and CNF $\clausetL(n)$ defined as follows.

\medskip

\hspace*{-0.5cm}
\begin{minipage}[t]{0.425\textwidth}
$\prefixL(n) := \forall B_1 \exists B_2 \forall B_3 \exists B_4$:

\begin{tabular}{ll}
$B_1 := \{ u_{3i + 1}, u_{3i + 2} \mid 0 \leq i < n\}$ & \\
$B_2 := \{ x_{3i + 1}, x_{3i + 2} \mid 0 \leq i < n \}$ & \\
\end{tabular}
\end{minipage}
\begin{minipage}[t]{0.55\textwidth}
\vphantom{$\prefixL(n) := \forall B_0 \exists B_1 \forall B_2 \exists B_3$:}

\begin{tabular}{ll}
$B_3 := \{ u_{3i + 3}  \mid 0 \leq i < n\}$ & \\
$B_4 := \{ x_{3i+3} \mid 0 \leq i < n \}$ & \\
\end{tabular}
\end{minipage}

\medskip

\hspace*{-0.5cm}
\begin{minipage}[t]{0.55\textwidth}
$\clausetL(n) := \bigwedge_{i := 0}^{n-1} \mathcal{C}(i)$ with $\mathcal{C}(i)
:= \bigwedge_{j := 0}^{7} C_{i,j}$:

\begin{tabular}{ll}
$C_{i,0} := $ & $(\neg u_{3i+2} \vee \neg x_{3i+1} \vee \neg x_{3i+2})$ \\
$C_{i,1} := $ & $(\neg u_{3i+1} \vee \neg x_{3i+1} \vee x_{3i+2}) $ \\
$C_{i,2} := $ & $(u_{3i+1} \vee x_{3i+1} \vee \neg x_{3i+2}) $ \\
$C_{i,3} := $ & $(u_{3i+2} \vee x_{3i+1} \vee x_{3i+2})  $  \\
\end{tabular}
\end{minipage}
\begin{minipage}[t]{0.425\textwidth}
\vphantom{$\clausetL(n) := \bigwedge_{i := 0}^{n-1}$}

\begin{tabular}{ll}
$C_{i,4} := $ & $(\neg x_{3i+1} \vee \neg x_{3i+2} \vee x_{3i+3}) $  \\
$C_{i,5} := $ & $(u_{3i + 3} \vee \neg x_{3i+3}) $ \\
$C_{i,6} := $ & $(\neg x_{3i+1} \vee x_{3i+2} \vee  \neg x_{3i+3})$ \\
$C_{i,7} := $ & $(\neg u_{3i + 3} \vee x_{3i+3})$ \\
\end{tabular}
\end{minipage}

\end{definition}

\begin{proposition}
\label{prop:stronger:qratu}
For $n \geq 1$, $\textnewqratu$ can eliminate the entire quantifier block $\forall B_1$ in $\FormulasL(n)$
whereas $\textqratu$ cannot eliminate any universal literals in $\FormulasL(n)$.
\end{proposition}
\begin{proof}[sketch, see appendix]
Formulas $\FormulasL(n)$ are constructed based on a similar principle
as $\FormulasC(n)$ in Definition~\ref{def:formula:class}.  
E.g., clauses $C_{i,0}$ and $C_{i,1}$ have $\newqrat$ but not $\qrat$ on
literals $\neg u_{3i+2}$ and $\neg u_{3i+1}$. During QUP, clauses $C_{i,5}$
and $C_{i,7}$ become empty only due to UR, which is not possible when using UP.
\qed
\end{proof}


\section{Proof Theoretical Impact of $\qrat$ and $\newqrat$} 
\label{sec:calculi}

As argued in the context of \emph{interference-based proof
systems}~\cite{DBLP:conf/cade/HeuleK17}, certain proof steps may become
applicable in a proof system only after redundant parts of the formula have
been eliminated.  We show that redundancy elimination by $\newqrat$ or $\qrat$
can lead to exponentially shorter proofs in the resolution based
\lquplusres~\cite{DBLP:conf/sat/BalabanovWJ14} QBF calculus.  Note that 
we do not compare the power of $\qrat$ or $\newqrat$ as proof systems
themselves, but  the impact of redundancy elimination on other proof
systems. The following result relies only on $\textqratu$, i.e., it does not require the
more powerful \mbox{redundancy property $\textnewqratu$ in $\newqrat$.}

\lquplusres is a calculus that generalizes
traditional Q-resolution~\cite{DBLP:journals/iandc/BuningKF95}. It allows to
generate resolvents on both existential and universal variables and admits
tautological resolvents of a certain kind. \lquplusres is among the strongest
resolution calculi currently
known~\cite{DBLP:conf/sat/BalabanovWJ14,beyersdorff_et_al:LIPIcs:2015:4905},
yet the following class of QBFs provides an exponential lower bound on the size of
\lquplusres proofs.

\begin{definition}[\cite{beyersdorff_et_al:LIPIcs:2015:4905}]
\label{def:lquplus:lower:bound}
For $n > 1$, let $\formulasQUParity(n) := \prefixQUParity(n). \clausetQUParity(n)$ be the
\emph{QUParity} QBFs with 
$\prefixQUParity(n) := \exists x_1, \ldots, x_n \forall z_1, z_2 \exists t_2,
\ldots, t_n$
and 
$\clausetQUParity(n) := C_0 \wedge C_1 \wedge \bigwedge_{i := 2}^{n} \mathcal{C}(i)$ 
where
$C_0 := (z_1 \vee z_2 \vee t_n)$, 
$C_1 := (\bar z_1 \vee \bar z_2 \vee \bar t_n)$, 
and 
$\mathcal{C}(i) := \bigwedge_{j := 0}^{7} C_{i,j}$:

\medskip

\begin{minipage}[t]{0.49\textwidth}
\begin{tabular}{ll}
$C_{2,0} := $ & $(\bar x_1 \vee \bar x_2 \vee z_1 \vee z_2 \vee \bar t_2)$ \\
$C_{2,1} := $ & $(x_1 \vee x_2 \vee  z_1 \vee z_2 \vee \bar t_2) $ \\
$C_{2,2} := $ & $(\bar x_1 \vee x_2 \vee z_1 \vee z_2 \vee t_2) $ \\
$C_{2,3} := $ & $(x_1 \vee \bar x_2 \vee z_1 \vee z_2 \vee t_2) $  \\
\end{tabular}
\begin{tabular}{ll}
$C_{2,4} := $ & $(\bar x_1 \vee \bar x_2 \vee \bar z_1 \vee \bar z_2 \vee \bar t_2)$ \\
$C_{2,5} := $ & $(x_1 \vee x_2 \vee \bar z_1 \vee \bar z_2 \vee \bar t_2) $ \\
$C_{2,6} := $ & $(\bar x_1 \vee x_2 \vee \bar z_1 \vee \bar z_2 \vee t_2) $ \\
$C_{2,7} := $ & $(x_1 \vee \bar x_2 \vee \bar z_1 \vee \bar z_2 \vee t_2) $  \\
\end{tabular}
\end{minipage}
\begin{minipage}[t]{0.49\textwidth}
\begin{tabular}{ll}
$C_{i,0} := $ & $(\bar t_{i-1} \vee \bar x_i \vee z_1 \vee z_2 \vee \bar t_i)$ \\
$C_{i,1} := $ & $(t_{i-1} \vee x_i \vee z_1 \vee z_2 \vee \bar t_i) $ \\
$C_{i,2} := $ & $(\bar t_{i-1} \vee x_i \vee z_1 \vee z_2 \vee t_i) $ \\
$C_{i,3} := $ & $(t_{i-1} \vee \bar x_i \vee z_1 \vee z_2 \vee t_i) $  \\
\end{tabular}
\begin{tabular}{ll}
$C_{i,4} := $ & $(\bar t_{i-1} \vee \bar x_i \vee \bar z_1 \vee \bar z_2 \vee \bar t_i)$ \\
$C_{i,5} := $ & $(t_{i-1} \vee x_i \vee \bar z_1 \vee \bar z_2 \vee \bar t_i) $ \\
$C_{i,6} := $ & $(\bar t_{i-1} \vee x_i \vee \bar z_1 \vee \bar z_2 \vee t_i) $ \\
$C_{i,7} := $ & $(t_{i-1} \vee \bar x_i \vee \bar z_1 \vee \bar z_2 \vee t_i) $  \\
\end{tabular}
\end{minipage}
\end{definition}

Any refutation of $\formulasQUParity(n)$ in \lquplusres is exponential in
$n$~\cite{beyersdorff_et_al:LIPIcs:2015:4905}. The QUParity formulas are a modification of the related
\emph{LQParity} formulas~\cite{beyersdorff_et_al:LIPIcs:2015:4905}. An
LQParity formula is obtained from a QUParity
formula $\formulasQUParity(n)$  by
replacing  $\forall z_1,z_2$ in prefix  $\prefixQUParity(n)$ by $\forall z$
and by replacing every occurrence of the literal pairs $z_1 \vee z_2$ and $\bar z_1
\vee \bar z_2$ in the clauses in $\clausetQUParity(n)$ by the literal $z$ and
$\bar z$, respectively.

\begin{proposition}
$\textqratu$ can eliminate either variable $z_1$ or $z_2$ from
  a QUParity formula $\formulasQUParity(n)$ to obtain a related LQParity formula in polynomial time.
\end{proposition}
\begin{proof}
We eliminate $z_2$ ($z_1$ can be eliminated alternatively) in a
polynomial number of $\textqratu$ steps. Every clause $C$ with $z_2
\in C$ has $\qrat$ on $z_2$ since $\{z_1, \bar z_1\} \subseteq
\outerres$ for all outer resolvents $\outerres$.  UP
immediately detects a conflict when propagating
$\overline{\outerres}$. After eliminating all literals $z_2$, the
clauses containing $\bar z_2$ trivially have $\qrat$ on $\bar z_2$,
which can be eliminated. Finally, $z_1$ including all of its occurrences is renamed to $z$.
\qed
\end{proof}

In the proof above the universal literals can be eliminated by $\textqratu$ in any
order. Hence in this case the non-confluence~\cite{DBLP:conf/nfm/HeuleSB15,DBLP:conf/cade/Kiesl017} of 
rewrite rules in the $\qrat$ and
$\newqrat$ systems is not an issue. \lquplusres has polynomial proofs for LQParity
formulas~\cite{beyersdorff_et_al:LIPIcs:2015:4905}. Hence the combination
of $\textqratu$  and
\lquplusres results in a calculus that is more powerful than \lquplusres. 
A related result~\cite{DBLP:conf/sat/KieslHS17} was obtained for the
combination of $\textqratu$ and
the weaker QU-resolution calculus~\cite{DBLP:conf/cp/Gelder12}.


\section{Experiments} \label{sec:experiments}

We implemented $\newqrat$ redundancy
removal in a tool called \qrattool for QBF preprocessing.\footnote{Source
  code of \qrattool: \url{https://github.com/lonsing/qratpreplus}} It applies rules $\textnewqrate$
and $\textnewqratu$ to remove redundant clauses and universal literals. We did not
implement clause addition ($\textnewqrata$) or extended universal reduction
($\texteur$). \qrattool is the
\emph{first implementation} of $\newqrat$ and $\qrat$ for QBF
preprocessing. The preprocessors 
\hqspre~\cite{DBLP:conf/tacas/WimmerRM017} and 
\bloqqer~\cite{DBLP:journals/jar/HeuleSB17} (which  generates partial $\qrat$
proofs to trace preprocessing steps) 
do not apply $\qrat$ to eliminate redundancies. 
The following experiments were run on a cluster of Intel Xeon CPUs
(E5-2650v4, 2.20 GHz) running Ubuntu 16.04.1. We  used the benchmarks from the PCNF
track of the QBFEVAL'17 competition. In terms of scheduling the
non-confluent (cf.~\cite{DBLP:conf/nfm/HeuleSB15,DBLP:conf/cade/Kiesl017}) rewrite rules $\textnewqrate$
and $\textnewqratu$, we have not yet optimized \qrattool. Moreover, in general large numbers of
clauses  in  formulas may 
cause run time overhead. In this
respect, our \mbox{implementation leaves room for improvements.}

\begin{table}[t]
\caption{Solved instances (\emph{S}), solved unsatisfiable (\emph{$\bot$}) and
satisfiable ones (\emph{$\top$}), and total wall clock time in kiloseconds (K) including time
outs on instances from QBFEVAL'17. Different combinations of
preprocessing by \bloqqer, \hqspre, and our tool \qrattool.  }
\addtocounter{table}{-1}
\begin{minipage}[b]{0.499\textwidth}
\begin{center}
\subfloat[Original instances (no prepr.).]{
{\setlength\tabcolsep{0.15cm}
\begin{tabular}{l@{\quad}c@{\quad}r@{\quad}r@{\quad}c}
\hline
\emph{Solver} & \emph{S} & \multicolumn{1}{l}{\emph{$\bot$}} & \emph{$\top$} & \emph{Time} \\
\hline
\caqe   & 170 & 128 & 42 &  656K  \\
\rareqs & 167 & 133 & 34 &  660K  \\
\depqbf & 152 & 108 & 44 &   690K \\
\qute   & 130 &  91 & 39 &  720K  \\
\hline
\end{tabular}
}
\label{fig:exp:523:instances:noprepro}
}
\end{center}
\end{minipage}
\hfill
\begin{minipage}[b]{0.499\textwidth}
\begin{center}
\subfloat[Prepr.~by \qrattool only.]{
{\setlength\tabcolsep{0.15cm}
\begin{tabular}{l@{\quad}c@{\quad}r@{\quad}c@{\quad}r}
\hline
\emph{Solver} & \emph{S} & \multicolumn{1}{l}{\emph{$\bot$}} & \emph{$\top$} & \emph{Time} \\
\hline
\caqe   & 209 & 141 & 68 & 594K   \\
\rareqs & 203 & 152 & 51 & 599K  \\
\depqbf & 157 & 109 & 48 & 689K \\
\qute   & 131 &  98 & 33 & 724K   \\
\hline
\end{tabular}
}
\label{fig:exp:523:instances:prepro}
}
\end{center}
\end{minipage}

\bigskip

\begin{minipage}[b]{0.499\textwidth}
\begin{center}
\subfloat[Prepr.~by \bloqqer only.]{
{\setlength\tabcolsep{0.15cm}
\begin{tabular}{l@{\quad}c@{\quad}c@{\quad}c@{\quad}c}
\hline
\emph{Solver} & \emph{S} & \emph{$\bot$} & \emph{$\top$} & \emph{Time} \\
\hline
\rareqs & 256 & 180 & 76 &   508K \\
\caqe   & 251 & 168 & 83  &  522K \\
\depqbf & 187 & 121 & 66 &   630K \\
\qute   & 154 & 109 & 45 &   682K \\
\hline
\end{tabular}
}
\label{fig:exp:523:instances:prepro:bloqqer}
}
\end{center}
\end{minipage}
\hfill
\begin{minipage}[b]{0.499\textwidth}
\begin{center}
\subfloat[Prepr.~by \bloqqer and \qrattool.]{
{\setlength\tabcolsep{0.15cm}
\begin{tabular}{l@{\quad}c@{\quad}c@{\quad}c@{\quad}r}
\hline
\emph{Solver} & \emph{S} & \emph{$\bot$} & \emph{$\top$} & \emph{Time} \\
\hline
\rareqs & 262 & 178 & 84 &   492K \\
\caqe   & 255 & 172 & 83 &   507K \\
\depqbf & 193 & 127 & 66 &   622K \\
\qute   & 148 & 107 & 41 &   688K \\
\hline
\end{tabular}
}
\label{fig:exp:523:instances:prepro:bloqqer:qrat}
}
\end{center}
\end{minipage}

\bigskip

\begin{minipage}[b]{0.499\textwidth}
\begin{center}
\subfloat[Prepr.~by \hqspre only.]{
{\setlength\tabcolsep{0.15cm}
\begin{tabular}{l@{\quad}c@{\quad}c@{\quad}c@{\quad}c}
\hline
\emph{Solver} & \emph{S} & \emph{$\bot$} & \emph{$\top$} & \emph{Time} \\
\hline
\caqe   & 306 & 197 & 109 & 415K \\
\rareqs & 294 & 194 & 100 & 429K \\
\depqbf & 260 & 171 &  89 & 494K \\
\qute   & 255 & 171 &  84 & 497K \\
\hline
\end{tabular}
}
\label{fig:exp:523:instances:prepro:hqspre}
}
\end{center}
\end{minipage}
\hfill
\begin{minipage}[b]{0.499\textwidth}
\begin{center}
\subfloat[Prepr.~by \hqspre and \qrattool.]{
{\setlength\tabcolsep{0.15cm}
\begin{tabular}{l@{\quad}c@{\quad}c@{\quad}c@{\quad}r}
\hline
\emph{Solver} & \emph{S} & \emph{$\bot$} & \emph{$\top$} & \emph{Time} \\
\hline
\caqe   & 314 & 200 & 114 & 407K \\
\rareqs & 300 & 195 & 105 & 418K  \\
\depqbf & 262 & 177 &  85 & 488K \\
\qute   & 250 & 169 &  81 & 500K \\
\hline
\end{tabular}
}
\label{fig:exp:523:instances:prepro:hqspre:qrat}
}
\end{center}
\end{minipage}
\label{fig:exp:523:instances}
\refstepcounter{table}
\end{table}
We illustrate the impact of QBF preprocessing by $\newqrat$ and $\qrat$ on
the performance of QBF solving. To this end, we
applied \qrattool in addition to the state of the art QBF preprocessors
\bloqqer and \hqspre. In the experiments,
first we preprocessed the benchmarks using \bloqqer and \hqspre, respectively, with a
generous limit of two hours wall clock time. We considered 39 and 42~formulas where
\bloqqer and \hqspre timed out, respectively, in their original form. 
Then we applied \qrattool to the preprocessed formulas with a \emph{soft wall
  clock time
  limit} of 600 seconds. When \qrattool reaches the limit, it
prints the formula with redundancies removed that
have been detected so far. These preprocessed formulas are then solved.  
Table~\ref{fig:exp:523:instances} shows the performance
of our solver \depqbf~\cite{DBLP:conf/cade/LonsingE17} in addition to the
top-performing solvers\footnote{We excluded the top-performing solver
\aigsolve due to observed assertion failures.} \rareqs~\cite{Janota20161},
\caqe~\cite{DBLP:conf/fmcad/RabeT15}, and \qute~\cite{DBLP:conf/sat/PeitlSS17}
from QBFEVAL'17, using limits of 7~GB and 1800~seconds wall clock time. The 
solvers implement different solving paradigms such as expansion or
resolution-based QCDCL. The
results clearly indicate the benefits of preprocessing by \qrattool. The
number of solved instances increases. \qute is an exception to this trend. We conjecture
that \qrattool blurs the formula structure in addition to
\bloqqer and \hqspre, which may be harmful to \qute.

We emphasize that we hardly observed a
difference in the effectiveness of redundancy removal by $\newqrat$ and $\qrat$ on the considered
benchmarks. The benefits of \qrattool shown in
Table~\ref{fig:exp:523:instances} are due to redundancy removal by $\qrat$
already, and not by $\newqrat$. However, on additional 672 instances from class
\emph{Gent-Rowley} (encodings of the Connect Four game) available from QBFLIB, 
$\textnewqrate$ on average removed 54\% more clauses than $\textqrate$. We
attribute this effect to larger numbers of quantifier blocks in the
Gent-Rowley instances (median 73, average 79) compared to QBFEVAL'17 (median
3, average 27). The advantage of QBF abstractions in the $\newqrat$ system 
is more pronounced on instances with many quantifier blocks. 


\section{Conclusion}
\label{sec:conclusion}

We presented $\newqrat$, a generalization of the $\qrat$ proof system,
that is based on a more powerful QBF redundancy property. The key
difference between the two systems is the use of QBF specific unit
propagation in contrast to propositional unit propagation. Due to
this, redundancy checking in $\newqrat$ is aware of quantifier
structures in QBFs, as opposed to $\qrat$. Propagation in $\newqrat$
potentially benefits from the presence of universal variables in the
underlying formula. This is exploited by the use of abstractions of
QBFs, for which we developed a theoretical framework, and from which
the soundness of $\newqrat$ follows.  By concrete classes of QBFs we
demonstrated that $\newqrat$ is more powerful than $\qrat$ in terms of
redundancy detection. Additionally, we reported on proof theoretical
improvements of a certain resolution based QBF calculus made by
$\qrat$ (or $\newqrat$) redundancy removal. A first experimental
evaluation illustrated the potential of redundancy elimination by
$\newqrat$.

As future work, we plan to implement a workflow for checking
$\newqrat$ proofs and extracting Skolem functions similar to $\qrat$
proofs~\cite{DBLP:journals/jar/HeuleSB17}. In our $\newqrat$
preprocessor \qrattool we currently do not apply a specific strategy
to handle the non-confluence of rewrite rules. We want to further
analyze the effects of non-confluence as it may have an impact on the
amount of redundancy detected. In our tool \qrattool we considered
only redundancy removal. However, to get closer to the the full power
of the $\newqrat$ system, it may be beneficial to also add redundant
clauses or universal literals to a formula.




\clearpage

\begin{appendix}

\section{Appendix}

\subsection{Example: Formula $\FormulasL(1)$}

The following example shows formula $\FormulasL(1)$ from the class $\FormulasL(n)$, which illustrates that $\textnewqratu$ is more powerful than $\textqratu$.

\begin{example}[related to Definition~\ref{def:formula:class:literal:elimination} on page~\pageref{def:formula:class:literal:elimination}]
For $n := 1$, we have $\FormulasL(n)$ with prefix $\prefixL(n) := \forall u_{1}, u_{2} \exists x_{1}, x_{2} \forall u_{3} \exists x_{3}$ and CNF  $\clausetL(n) := \mathcal{C}(0)$ with $\mathcal{C}(0)
:= \bigwedge_{j := 0}^{7} C_{0,j}$ as follows.

\hspace*{-0.5cm}
\begin{minipage}[t]{0.55\textwidth}

\begin{tabular}{ll}
$C_{0,0} := $ & $(\neg u_{2} \vee \neg x_{1} \vee \neg x_{2})$ \\
$C_{0,1} := $ & $(\neg u_{1} \vee \neg x_{1} \vee x_{2}) $ \\
$C_{0,2} := $ & $(u_{1} \vee x_{1} \vee \neg x_{2}) $ \\
$C_{0,3} := $ & $(u_{2} \vee x_{1} \vee x_{2})  $  \\
\end{tabular}
\end{minipage}
\begin{minipage}[t]{0.425\textwidth}

\begin{tabular}{ll}
$C_{0,4} := $ & $(\neg x_{1} \vee \neg x_{2} \vee x_{3}) $  \\
$C_{0,5} := $ & $(u_{3} \vee \neg x_{3}) $ \\
$C_{0,6} := $ & $(\neg x_{1} \vee x_{2} \vee  \neg x_{3})$ \\
$C_{0,7} := $ & $(\neg u_{3} \vee x_{3})$ \\
\end{tabular}
\end{minipage}

\end{example}

\newtheorem{innercustomlem}{Lemma}
\newenvironment{customlem}[1]
  {\renewcommand\theinnercustomlem{#1}\innercustomlem}
  {\endinnercustomlem}

  \newtheorem{innercustomthm}{Theorem}
\newenvironment{customthm}[1]
  {\renewcommand\theinnercustomthm{#1}\innercustomthm}
  {\endinnercustomthm}

    \newtheorem{innercustomprop}{Proposition}
\newenvironment{customprop}[1]
  {\renewcommand\theinnercustomprop{#1}\innercustomprop}
  {\endinnercustomprop}

\subsection{Proofs}

\noindent\textbf{Proof of Lemma~\ref{lem:outerres:satisfied} on page~\pageref{lem:outerres:satisfied}:}

\begin{customlem}{5}
  Given a clause $C$ with $\newqior$  with respect to QBF
$\qclauset := \prefix. \clauset$ on literal $l \in C$ with $\var{l} = x$. Let 
$T$ be a model of $\qclauset$ and  $\tau \subseteq T$ be a path in $T$. If
$\tau(C \setminus \{l\}) =
\bot$ then $\tau^x(D) = \top$ for all $D \in \clauset$ \mbox{with $\bar l \in D$.}
\end{customlem}

\begin{proof}[similar to proof of Lemma~6 in~\cite{DBLP:journals/jar/HeuleSB17}]
Let $D \in \clauset$ be a clause with $\bar l \in D$ and consider $R :=
\outerres(\prefix,C,D,l) = (C \setminus \{l\}) \cup \outercl(\prefix, D, \bar
l)$.  By Definition~\ref{def:new:qior}, we have $\prefix. \clauset
\treemodelequiv \prefix. (\clauset \wedge \outerres(\prefix, C, D, l))$ for
all $D \in \clauset$ with $\bar l \in D$.  Let $T$ be a model of
$\prefix. \clauset$ and $\tau \subseteq T$ a path in $T$. Since $T \qbfmodels
\prefix. \clauset$ and $T \qbfmodels \prefix. (\clauset \wedge \outerres(\prefix, C, D,
l))$, we have $\tau(\clauset) = \top$ and $\tau(R) = \top$.  Assuming that
$\tau(C \setminus \{l\}) = \bot$, we have 
$\tau(\outercl(\prefix, D, \bar l)) = \top$ since $\tau(R) = \top$. The clause
$\outercl(\prefix, D, \bar l)$ does not contain $\bar l$ and it contains only
literals of variables from blocks smaller than or equal to the block
containing $x$. Hence we have $\tau^x(\outercl(\prefix, D, \bar l)) = \top$
for the prefix assignment $\tau^x$, and further $\tau^x(D) = \top$ since
$\outercl(\prefix, D, \bar l) \subseteq D$.
\qed
\end{proof}

\noindent\textbf{Proof of Theorem~\ref{thm:newqior:clause:soundness} on page~\pageref{thm:newqior:clause:soundness}:}

\begin{customthm}{3}
  Given a QBF $\qclauset := \prefix. \clauset$ and a clause $C \in \clauset$ with $\newqior$ on
an \emph{existential} literal $l \in C$ with respect to QBF
$\qclauset' := \prefix'. \clauset'$ where $\clauset' := \clauset \setminus \{C\}$ and
$\prefix'$ is the same as $\prefix$ with variables and respective quantifiers
removed that no longer appear in $\clauset'$. Then $\qclauset \satequiv
\qclauset'$. 
\end{customthm}

\begin{proof}[similar to proof of Theorem~7
    in~\cite{DBLP:journals/jar/HeuleSB17}]
We can adapt the prefix $\prefix'$ of $\qclauset'$ to be the same as the prefix of $\qclauset$ in a satisfiability-preserving way. 
  If $\qclauset$ is satisfiable then
$\qclauset'$ is also satisfiable since every model of $\qclauset$ is also a
model of $\qclauset'$.  Let $T'$ be a model of $\qclauset'$ and $T^P := T'$ 
a pre-model of $\qclauset$. Consider paths $\tau \subseteq T^P$ in $T^P$ for
which we have $\tau(\clauset') = \top$ but $\tau(C) = \bot$, where $\tau =
\tau^x\bar l\tau_x$ for $\var{l} = x$ and $l \in C$. Since $\tau(C) = \bot$ 
also $\tau(C \setminus \{l\}) = \bot$, and due to
Lemma~\ref{lem:outerres:satisfied} we have $\tau^x(D) = \top$ for all $D \in
\clauset$ with $\bar l \in D$.  We construct a pre-model $T$ of $\qclauset$
from $T^P$ by modifying \emph{all} such paths $\tau \subseteq T^P$ by flipping the
assignment of $x$ to obtain $\tau' := \tau^xl\tau_x$ such that $\tau' \subseteq
T$.
(If we process multiple redundant clauses $C$, then cyclic
modifications by assignment flipping cannot occur if we do the
modifications in reverse chronological ordering as the clauses were
detected redundant. This is the same principle of reconstructing
solutions when using blocked clause elimination, for example.)
Now $\tau'(C) = \top$ and also $\tau'(D) = \top$ since
$\tau^x(D) = \top$, and $\tau$ and $\tau'$ have the same prefix assignment
$\tau^x$. Hence $T \qbfmodels \qclauset$ and thus $\qclauset$ is satisfiable.
\qed
\end{proof}

\noindent\textbf{Proof of Theorem~\ref{thm:newqior:literal:soundness} on page~\pageref{thm:newqior:literal:soundness}:}

\begin{customthm}{4}
  Given a QBF $\qclauset_0 := \prefix. \clauset$ and $\qclauset :=
\prefix. (\clauset \cup \{C\})$ where $C$ has $\newqior$ on a \emph{universal}
literal $l \in C$ with respect to $\qclauset_0$. Let $\qclauset' :=
\prefix. (\clauset \cup \{C'\})$ with $C' := C \setminus \{l\}$. Then $\qclauset
\satequiv \qclauset'$.

\end{customthm}

\begin{proof}[similar to proof of Theorem~8
in~\cite{DBLP:journals/jar/HeuleSB17}] If $\qclauset'$ is satisfiable then
$\qclauset$ is also satisfiable since every model of $\qclauset'$ is also a
model of $\qclauset$.  Let $T$ be a model of $\qclauset$ and $T^P := T$ be a
pre-model of $\qclauset'$. Consider paths $\tau \subseteq T^P$ in $T^P$ for
which we have $\tau(\clauset) = \top$ and $\tau(C) = \top$ but $\tau(C') =
\bot$. Since $C' = C \setminus \{l\}$, we have $\tau = \tau^xl\tau_x$ for
$\var{l} = x$ and $l \in C$.  Since $l$ is universal, for every such $\tau$
there exists a path $\tau' \subseteq T^P$ with $\tau' = \tau^x\bar l\rho_x$,
with $\tau_x$ and $\rho_x$ being different suffix assignments of $\tau$ and
$\tau'$, respectively. We have $\tau'(\clauset) = \top$ and $\tau'(C) = \top$
since $\tau' \subseteq T$ because $T^P = T$, and also $\tau'(C') = \top$
because $l \in C$ but $\bar l \in \tau'$. Hence $C'$ is satisfied by $\tau'$
due to some assignment $k \in \rho_x$.  Due to $\tau(C') = \tau(C \setminus
\{l\}) = \bot$ and Lemma~\ref{lem:outerres:satisfied} we have $\tau^x(D) =
\top$ for all $D \in \clauset$ with $\bar l \in D$ and hence also $\tau'(D) =
\top$ since $\tau$ and $\tau'$ have the same prefix assignment $\tau^x$.  We
construct a pre-model $T'$ of $\qclauset'$ from $T^P$ by modifying \emph{all}
paths $\tau = \tau^xl\tau_x$ for which $\tau(C') = \bot$ to be $\tau'' :=
\tau^xl\rho_x$, where $\rho_x$ is the suffix assignment of path $\tau' =
\tau^x\bar l\rho_x$ that corresponds to the other branch $\bar l$ of the
universal literal $l$. These modifications in fact are a replacement of the
subtree under $\tau^xl$.
(As noted in the proof of Theorem~\ref{thm:newqior:clause:soundness}
above, cyclic modifications cannot occur if we process multiple
redundant clauses $C$, provided that we do the modifications in
reverse chronological ordering as the clauses were detected
redundant.)
We have $\tau''(C') = \top$ due to its suffix
assignment $\rho_x$, and also $\tau''(\clauset) = \top$. Therefore, $T'
\qbfmodels \qclauset'$ and hence $\qclauset'$ is satisfiable.  \qed
\end{proof}

\noindent\textbf{Extended Proof Sketch of Proposition~\ref{prop:stronger:qratu} on page~\pageref{prop:stronger:qratu}:}

\begin{customprop}{5}
  For $n \geq 1$, $\textnewqratu$ can eliminate the entire quantifier block $\forall B_1$ in $\FormulasL(n)$
whereas $\textqratu$ cannot eliminate any universal literals in $\FormulasL(n)$.
\end{customprop}

\begin{proof}[sketch]
Formulas $\FormulasL(n)$ are constructed based on a similar principle
as $\FormulasC(n)$ in Definition~\ref{def:formula:class}.  For $i$ and
$k$ with $i \not = k$, the sets of variables in $\mathcal{C}(i)$ and
$\mathcal{C}(k)$ are disjoint.  Thus it suffices to prove the claim
for an arbitrary $\mathcal{C}(i)$.  Clause $C_{i,0}$ has $\newqrat$ on
literal $\neg u_{3i+2}$. The relevant outer resolvent is
$\outerres_{0,3} = \outerres(\prefixL(n), C_{i,0}, C_{i,3}, \neg
u_{3i+2}) = (\neg x_{3i+1} \vee \neg x_{3i+2})$. We have
$\max(\levels(\outerres_{0,3})) = 2$, and variable $u_{3i + 3}$ is
universal in $\eabs(\FormulasL(n), 2)$. Propagating
$\overline{\outerres_{0,3}}$ makes $C_{i,4}$ unit, finally $C_{i,5}$
becomes empty under the derived assignment $x_{3i + 3}$ and since UR
reduces the literal $u_{3i + 3}$. After removing literal $\neg
u_{3i+2}$ from $C_{i,0}$, clause $C_{i,3}$ trivially has $\newqrat$ on
$u_{3i+2}$. The literals of variable $u_{3i+1}$ in $C_{i,1}$ and
$C_{i,2}$ can be eliminated in a similar way, where clause $C_{i,7}$
becomes empty by UR in QUP. 
In contrast to $\textnewqratu$, $\textqratu$ cannot eliminate any
universal literals in $\FormulasL(n)$.  Clauses $C_{i,5}$ and
$C_{i,7}$ in $\FormulasL(n)$ do not become empty. All variables are
existential since UP is applied.  The claim can be proved by case
analysis of all possible outer resolvents.
\qed
\end{proof}


\end{appendix}

\end{document}